\documentclass{article}

\usepackage{PRIMEarxiv}

\usepackage[utf8]{inputenc} 
\usepackage[T1]{fontenc}    
\usepackage{hyperref}       
\usepackage{url}            
\usepackage{booktabs}       
\usepackage{amsfonts}       
\usepackage{nicefrac}       
\usepackage{microtype}      
\usepackage{lipsum}
\usepackage{fancyhdr}       
\usepackage{graphicx}       

\usepackage[export]{adjustbox}

\usepackage[english]{babel}

\usepackage{doi}

\usepackage{comment}  

\usepackage{amsmath} 

\usepackage{amssymb} 

\usepackage{amsthm}

\usepackage[ruled,linesnumbered]{algorithm2e}

\usepackage{float}

\usepackage{subcaption}

\usepackage{silence}

\usepackage{dirtytalk}

\newtheorem{theorem}{Theorem}[section]
\newtheorem{lemma}[theorem]{Lemma}

\graphicspath{{media/}}     

\title{Optimizing Communication in Byzantine Agreement Protocols with Slim-HBBFT}

\author{ {Nasit S Sony} \\
	University of California, Merced\\
	CA 95340, USA \\
	\texttt{nsony@ucmerced.edu} \\
	\And
	{Xianzhong Ding} \\
	Lawrence Berkeley National Laboratory\\
	CA 94720, USA \\
	\texttt{dingxianzhong@lbl.gov} \\
        \And
	{Mukesh Singhal} \\
	University of California, Merced\\
	CA 95340, USA \\
	\texttt{msinghal@ucmerced.edu} \\
}

\begin{document}
\maketitle

\begin{abstract}

Byzantine agreement protocols in asynchronous networks have received renewed interest because they do not rely on network behavior to achieve termination. Conventional asynchronous Byzantine agreement protocols require every party to broadcast its requests (e.g., transactions), and at the end of the protocol, parties agree on one party's request. If parties agree on one party's requests while exchanging every party's request, the protocol becomes expensive. These protocols are used to design an atomic broadcast (ABC) protocol where parties agree on $\langle n-f \rangle$ parties' requests (assuming $n=3f+1$, where $n$ is the total number of parties, and $f$ is the number of Byzantine parties). Although the parties agree on a subset of requests in the ABC protocol, if the requests do not vary (are duplicated), investing in a costly protocol is not justified. We propose Slim-HBBFT, an atomic broadcast protocol that considers requests from a fraction of $n$ parties and improves communication complexity by a factor of $O(n)$. At the core of our design is a prioritized provable-broadcast (P-PB) protocol that generates proof of broadcast only for selected parties. We use the P-PB protocol to design the Slim-HBBFT atomic broadcast protocol. Additionally, we conduct a comprehensive security analysis to demonstrate that Slim-HBBFT satisfies the properties of the Asynchronous Common Subset protocol, ensuring robust security and reliability.

\end{abstract}


\section{Introduction}
\label{sec:intro}

Byzantine agreement (BA) is a fundamental problem in computer systems, first introduced by Lamport, Pease, and Shostak in their pioneering works \cite{BYZ22,BYZ23}. The problem assumes a system where multiple computers, referred to as machines, parties, or nodes, try to agree on a value despite some of these computers being Byzantine (i.e., behaving arbitrarily, unpredictably, or maliciously). Since the seminal work \cite{BYZ23}, numerous models with various system assumptions have been proposed to solve the BA problem.


Bitcoin \cite{BITCOIN01} has refueled interest in Byzantine agreement protocols, particularly in asynchronous networks where the protocol's independence from time parameters is crucial. In asynchronous networks, each party must broadcast its requests. The atomic broadcast (ABC) protocol is designed to order these requests. In an ABC protocol, parties agree on $\langle n-f \rangle$ common requests, where $\langle f+1 \rangle$ of these requests come from honest parties. Although honest parties may propose varied requests, they might still broadcast the same requests due to different orderings of client requests and lack of knowledge of other parties' requests until an agreement is reached. Consequently, agreeing on a subset of requests does not necessarily improve the total number of accepted requests. We follow the committee approach like \cite{cMVBA,PMVBA,sony2025efficient,SlimABC,KSizeVABA,OHBBFT}.


Our main observations are that reducing the number of proposals leads to a more efficient protocol. We follow the committee approach from \cite{cMVBA,PMVBA,sony2025efficient,SlimABC,KSizeVABA,OHBBFT}. We leverage this reduction technique to design an atomic broadcast protocol. We introduce a slim atomic broadcast protocol where parties agree on $q$ $(1 \leq q \leq f+1)$ requests. This protocol reduces communication complexity by $O(n)$. We randomly select $\langle f+1 \rangle$ parties to broadcast their requests/proposals, ensuring at least one honest party is included, with an average of $\frac{2}{3}$ of the selected parties being honest. Consider the following two cases:
\begin{enumerate}
    \item If parties agree on one party's request, the communication cost is lower regardless of request variation among selected parties.
    \item If parties agree on $q$ proposals with non-varying requests, the protocol maintains low communication cost. If requests vary among the $q$ parties, the protocol benefits from both reduced communication cost and an increased number of accepted requests.
\end{enumerate}

\section{Slim-HBBFT}

\begin{figure}[t]
    \centering
    \includegraphics[width=0.792\textwidth,height=0.252\textwidth]{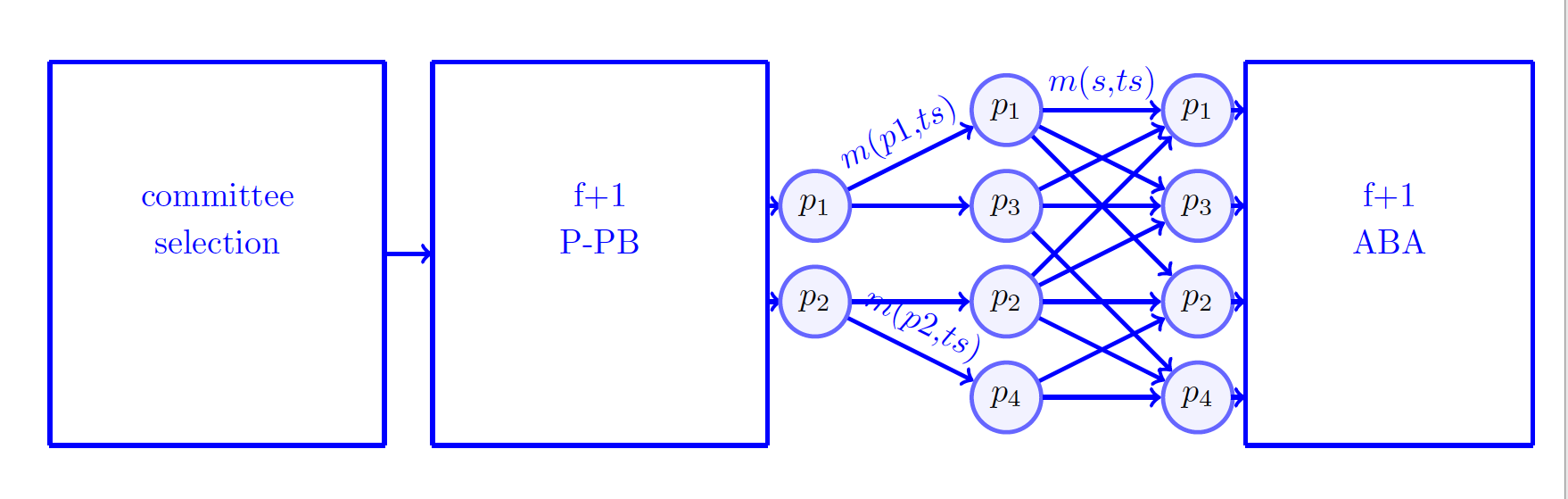}
    \caption{Slim-HBBFT illustration. Each selected party promotes its request using the P-PB protocol. After completing the P-PB protocol, a party broadcasts the requests with proof $ts$ as a proposal. Upon receiving a proposal, each party suggests it and waits for $\langle n-f \rangle$ suggestions, then runs the $ABA$ protocol as input.}    
    \label{fig:SlimACS}
\end{figure}
\paragraph{\textbf{High-Level overview.}}
Slim-HBBFT is an atomic broadcast protocol where honest parties agree on $q$ ($1 \leq q \leq f+1$) parties' requests.


\paragraph{\textbf{Committee Selection}} We select $\kappa = f+1$ parties using a standard committee selection protocol. These selected parties propose their requests (value $v$) using the P-PB protocol. The communication complexity of $n$ P-PB instances is $O(n^2v)$, where $v < Kn^2 \log n$ and $K$ is the security parameter. In contrast, the communication complexity of $n$ RBC instances is $O(n^3v)$. HoneyBadgerBFT \cite{HONEYBADGER01} uses erasure coding to broadcast requests, reducing the communication complexity to $O(n^2 |v| + Kn^3 \log n)$. By using the P-PB protocol, we eliminate the $O(Kn^3 \log n)$ term. However, the P-PB protocol does not guarantee the totality property, which ensures that all parties receive a selected party's broadcast. To address this, we introduce two additional steps: propose and suggest.

\paragraph{\textbf{Propose and Suggest Steps}}
After completing the P-PB protocol, a selected party broadcasts the proof as a proposal (propose step). Upon receiving a proposal, each party suggests the proposal (suggest step) and waits for $\langle n-f \rangle$ suggestions. If a party receives $\langle n-f \rangle$ suggestions, it implies receiving multiple suggestions for at least one proposal (total proposals are $\langle f+1 \rangle$, and a party receives $\langle 2f+1 \rangle$ proposals). If a proposal is suggested by multiple parties, it reaches more parties. This increases the likelihood more honest parties receive the proposal. An ABA instance outputs $1$ if an honest party inputs $1$. Each party invokes the ABA protocol for each selected party and inputs $1$ if it receives a proposal or suggestion. After receiving $\langle n-f \rangle$ suggestions, an honest party inputs $0$ to an ABA instance if there are no suggestions for the selected ABA instance.

\paragraph{\textbf{Threshold Encryption}}

To prevent adversaries from censoring broadcasts, we use a threshold encryption scheme. A party encrypts its requests using the threshold encryption scheme and then broadcasts the encrypted request. Decrypting the message requires $\langle f+1 \rangle$ parties' decryption shares, and an honest party reveals its decryption share after reaching an agreement. This ensures adversaries cannot decrypt an honest party's proposal until an agreement is reached. Although the encryption scheme adds $O(Kn^2)$ communication bits, this is negligible compared to the protocol's overall complexity. We ensure each proposal comes from a selected party that completes the P-PB protocol before parties input to the ABA instance. The P-PB protocol ensures the request is sent among $\langle f+1 \rangle$ honest parties, who then add their sign-share to the request. Figure \ref{fig:SlimACS} provides an illustration of the Slim-HBBFT protocol.

\section{Security Analysis}
 Slim-HBBFT provides an atomic broadcast protocol for a subset of parties' requests by utilizing the Asynchronous Common Subset (ACS) protocol and a threshold encryption scheme. To analyze the security of the Slim-HBBFT protocol, we demonstrate that it satisfies all the properties of the ACS protocol. A key requirement of the ACS protocol is that at least one of the provable-broadcasts (threshold-signature) reaches $2f+1$ parties. The following lemmas support this requirement.

\begin{lemma}\label{one proposal}
    In the $propose$ step, at least one proposal reaches multiple parties.

\end{lemma} 

\begin{proof}
  
    We know that $\langle f+1 \rangle$ parties propose their requests, and $\langle 3f+1 \rangle$ parties receive at least one proposal. Therefore, due to the fraction $\frac{3f+1}{f+1}$, at least one proposal is common to more than one party.
\end{proof} 

\begin{lemma} \label{common proposal}

    In the suggestion step, at least one proposal reaches $\langle 2f+1 \rangle$ parties.
    
\end{lemma} 

 \begin{proof}
   
We know that $\langle 2f+1 \rangle$ honest parties suggest their received proposals. From Lemma \ref{one proposal}, at least one proposal is suggested by more than one party. Assume no proposal is common to more than $2f$ parties. However, every party waits for $\langle 2f+1 \rangle$ suggestions, and there must be $\langle 2f+1 \rangle \times \langle 2f+1 \rangle$ suggestions. If no proposal is suggested to more than $2f$ parties, the total number of suggestions would be $\langle 2f+1 \rangle \times \langle 2f \rangle < \langle 2f+1 \rangle \times \langle 2f+1 \rangle$. Since honest parties must send enough suggestions to ensure protocol progress, at least one party's proposal must be received by $\langle 2f+1 \rangle$ parties.   
\end{proof}

\begin{theorem}
With negligible probability of failure, the Slim-HBBFT protocol satisfies the Agreement, Validity, and Totality properties of the ACS protocol, given the security of the underlying Prioritized Provable Broadcast, Committee Selection, and ABA protocols.
\end{theorem}
\section{Conclusion}

This work introduces Slim-HBBFT, a novel atomic broadcast protocol designed to enhance the efficiency of Byzantine agreement in asynchronous networks. The key challenge addressed is the high communication complexity of traditional Byzantine agreement protocols, especially when dealing with duplicate requests from different parties. By leveraging a committee-based approach and the Prioritized Provable-Broadcast protocol, Slim-HBBFT significantly reduces the number of proposals needed for agreement, thus lowering communication costs. Our analytical performance analysis demonstrates that Slim-HBBFT can efficiently manage communication complexity while ensuring at least one honest party's proposal is included in the agreement. The protocol's design also prevents adversarial censorship through the use of threshold encryption, ensuring security and reliability. Future work will involve simulating the Slim-HBBFT protocol to explore optimal parameter ranges and develop a dynamic switching mechanism for selecting efficient committee sizes based on application needs, urgency of requests, and node patterns. Additionally, we will conduct comprehensive security and efficiency analyses to validate the protocol's robustness and performance.



\bibliography{references}

\appendix

\end{document}